\documentclass[11pt]{article}

\usepackage{textcomp}

\usepackage{amsmath}
\usepackage{amssymb}
\usepackage{amsthm}
\usepackage{todonotes}
\usepackage{tikz}
\usepackage{pgfplots}
\usepackage{graphicx}
\usepackage{subfigure}
\usepackage{epstopdf}
\usepackage[labelfont=bf,font=footnotesize]{caption}
\usepackage[labelfont=bf,font=footnotesize]{subcaption}
\usepackage{xcolor}
\usepackage{bbm}
\usepackage{makecell}
\usepackage{authblk}
\usepackage{fullpage,etoolbox}
\usepackage[colorlinks=true,citecolor=blue,linkcolor=blue,urlcolor=blue]{hyperref}
\usepackage[shortlabels]{enumitem}
\usepackage[round,sort,compress]{natbib}
\usepackage{listings}
\usepackage{commands}
\usepackage[ruled,vlined,linesnumbered]{algorithm2e}

\usepackage{balance}  

\title{
Coordinating Planning and Tracking in Layered Control Policies via Actor-Critic Learning
}

\author{Fengjun Yang and Nikolai Matni
\thanks{F. Yang is with the Dept. of Comput. and Info. Sci,
        University of Pennsylvania, PA, USA. N. Matni is with the Dept. of Elect. and Syst. Eng., University of Pennsylvania, PA, USA. This work is supported in part by NSF award ECCS-2231349, SLES-2331880, and CAREER-2045834.}%
}

\begin{document}

\maketitle
\thispagestyle{empty}
\pagestyle{empty}

\begin{abstract}

We propose a reinforcement learning (RL)-based algorithm to jointly train (1) a trajectory planner and (2) a tracking controller in a layered control architecture. Our algorithm arises naturally from a rewrite of the underlying optimal control problem that lends itself to an actor-critic learning approach. By explicitly learning a \textit{dual} network to coordinate the interaction between the planning and tracking layers, we demonstrate the ability to achieve an effective consensus between the two components, leading to an interpretable policy. We theoretically prove that our algorithm converges to the optimal dual network in the Linear Quadratic Regulator (LQR) setting and empirically validate its applicability to nonlinear systems through simulation experiments on a unicycle model.

\end{abstract}

\section{Introduction}
Layered control architectures \citep{matni2024towards, chiang2007layering} are ubiquitous in complex cyber-physical systems, such as power networks, communication networks, and autonomous robots. For example, a typical autonomous robot has an autonomy stack consisting of decision-making, trajectory optimization, and low-level control. However, despite the widespread presence of such layered control architectures, there has been a lack of a principled framework for their design, especially in the data-driven regime.

In this work, we propose an algorithm for jointly learning a trajectory planner and a tracking controller. We start from an optimal control problem and show that a suitable relaxation of the problem naturally decomposes into reference generation and trajectory tracking layers. We then propose an algorithm to train a layered policy parameterized in a way that parallels this decomposition using actor-critic methods. Different from previous methods, we show how a \textit{dual network} can be trained to coordinate the trajectory optimizer and the tracking controller. Our theoretical analysis and numerical experiments demonstrate that the proposed algorithm can achieve good performance in various settings while enjoying inherent interpretability and modularity.

\subsection{Related Work}
\subsubsection{Layered control architectures} The idea of layering has been studied extensively in the multi-rate control literature \citep{rosolia2022unified, csomay2022multi}, through the lens of optimization decomposition \citep{chiang2007layering, matni2016theory}, and for specific application domains \citep{samad2007system, samad2017controls, jiang2018improved}. Recently, Matni et al. \citep{matni2024towards} proposed a quantitative framework for the design and analysis of layered control architectures, which has since been instantiated to various control and robotics applications \citep{srikanthan2023augmented, srikanthan2023data, zhang2024change}. Within this framework, our work is most related to \citet{srikanthan2023data, zhang2024change}, which seek to design trajectory planners based on past data of a tracking controller. However, we consider the case where the low-level tracking controller is not given and has to be learned with the trajectory planner. We also provide a more principled approach to coordinating planning and tracking that leverages a dual network.

\begin{figure}
    \centering
    \subfigure[Without Dual Learning]{
        \includegraphics[width=0.4\textwidth]{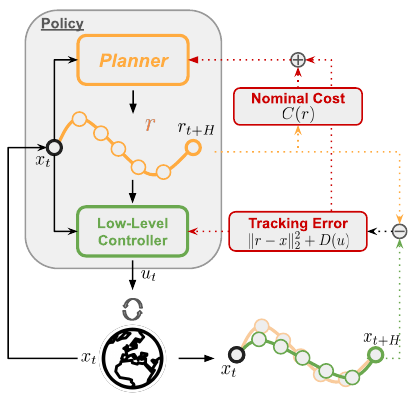}
        \label{fig:schematic-without-dual}
    }
    \hspace{0.05\textwidth}
    \subfigure[With Dual Learning (\textbf{Ours})]{
        \includegraphics[width=0.4\textwidth]{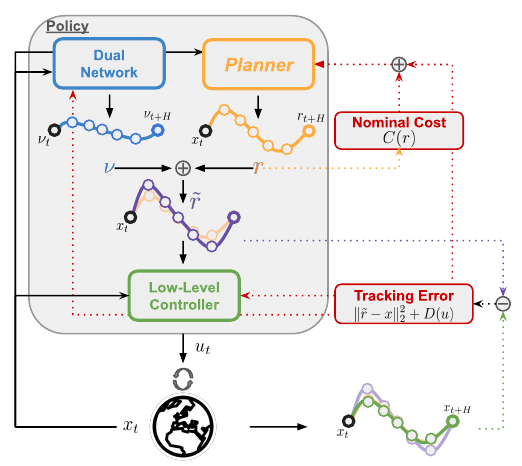}
        \label{fig:schematic-with-dual}
    }
    \caption{Comparison of trajectory planning and tracking approaches. (a) Previous approaches integrate a trajectory planner and a low-level controller by feeding the reference trajectory $r$ generated by the planner directly into the low-level controller. The low-level tracking controller minimizes the tracking cost, while the planner minimizes both the tracking cost and the nominal cost $\mathcal{C}(r)$. However, due to the tracking controller's imperfections, the executed trajectory often deviates from the reference, resulting in suboptimal performance. (b) Our proposed method introduces an additional dual network that learns to preemptively perturb the reference trajectory $r$ to $\tilde{r}$, accounting for the low-level controller's inaccuracies. By perturbing the reference trajectory, the executed trajectory $x$ is closer to the actual reference $r$, thus improving overall performance. We show that this module can be trained in the fashion of a dual update (hence the name) by observing the discrepancy between the reference and the executed trajectory.}
    \label{fig:schematics}
\end{figure}

\subsubsection{Hierarchical reinforcement learning}
Recently, reinforcement learning-based methods have demonstrated impressive performance on highly complex dynamical systems \citep{kumar2021rma, kaufmann2023champion}. Within the RL literature, our approach is most closely related to the idea of \textit{goal-conditioned} reinforcement learning \citep{dayan1992feudal, kulkarni2016hierarchical, levy2017learning, nachum2018data, vezhnevets2017feudal, nachum2018near}. In this framework, an upper-level agent periodically specifies a goal for the lower-level agent to execute. However, the ``intrinsic" reward used to train the lower-level agent is usually heuristically chosen. Nachum et al.~\citep{nachum2018near} derived a principled objective for the lower-level agent based on a suboptimality bound introduced by the hierarchical structure, but they focus on the case where the goal is specified as a learned low-dimensional representation. We focus on the case where the dynamics are deterministic and derive a simple quadratic objective for the lower-level agent (tracking layer). We also structure our upper-level agent (planning layer) to generate full trajectories instead of single waypoints.

\subsubsection{Actor-critic methods}
The actor-critic method \citep{silver2014deterministic, lillicrap2015continuous, fujimoto2018addressing} describes a class of reinforcement learning algorithms that simultaneously learn a policy and its associated value function. These algorithms have achieved great success with continuous control tasks and have found various applications in the controls and robotics community \citep{wang2024actor, grandesso2023cacto}. In this paper, we use actor-critic methods to learn a tracking controller and its value function, where the latter is used to help the trajectory planner determine how difficult a generated trajectory is for the tracking controller to follow.

\subsection{Statement of Contributions}
Our contribution is three-fold. First, we propose a novel way of parameterizing layered policies based on a principled derivation. In this parameterization, we introduce a \textit{dual network} to coordinate the trajectory planner and the tracking controller. We show how this dual network can be trained jointly with other components in the layered policy in an RL fashion. Secondly, we show theoretically and empirically that our algorithm for updating the dual network can recover the optimal dual network parameters for unconstrained linear quadratic regulator (LQR) problems. Finally, we evaluate our method empirically on constrained LQR problems and the unicycle environment to demonstrate its potential to be applied to more complex systems.

\section{Problem Formulation}\label{sec:problem-formulation}
We consider a discrete-time finite-horizon optimal control problem with state $x_t \in \mathbb{R}^{d_x}$ and control input $u_t \in \mathbb{R}^{d_u}$:
\begin{equation}\label{eq:original-prob}
    \begin{aligned}
        \underset{x_{0:T}, u_{0:T-1}}{\text{minimize}}&\quad \E_{\xi \sim D_\xi}
            \left[ \mathcal{C}(x_{0:T}) + \mathcal{D}(u_{0:T-1}) \right]\\
        \text{subject to}&\quad x_{t+1} = f(x_t, u_t), \quad \forall t=0, 1, ..., T-1,\\
        &\quad x_{0:T}\in \mathcal{X}, \quad u_{0:T-1} \in \mathcal{U}, 
        \quad x_0 = \xi.
    \end{aligned}
\end{equation}
Here, $T \in \mathbb{Z^+}$ is a fixed time horizon, $x_{0:T}=[x_0^\top, ..., x_T^\top]^\top$ and $u_{0:T-1}=[u_0^\top, ..., u_{T-1}^\top]^\top$ respectively denote the state and control trajectory. $\mathcal{C}(x_{0:T})$ and $\mathcal{D}(u_{0:T-1})$ are the state and control costs, respectively. We assume that the input cost and the state and input constraints decouple across time, and denote them respectively by ${\mathcal{D}(u_{0:T-1}) = \sum_{t=0}^{T-1}\mathcal{D}_t(u_t)}$, $\mathcal{X} = \prod_{t=0}^T\mathcal{X}_t$, and $\mathcal{U}=\prod_{t=0}^T\mathcal{U}_t$. The initial condition $\xi$ is sampled i.i.d. from a possibly unknown distribution $D_\xi$.

As per the reinforcement learning convention, we assume that we only have access to the dynamics via a simulator, i.e., that we do not know $f(x_t, u_t)$ explicitly, but can simulate the dynamics for any $x_t$ and $u_t$. However, we do assume that we have access to the cost functions $\mathcal{C}$, $\mathcal{D}$, as they are usually designed by the users, instead of being an inherent hidden part of the system. We also assume that we know the constraints $\mathcal{X}$ and $\mathcal{U}$ for the same reason.


Our goal is to learn a layered policy ${\pi=(\pi^\mathrm{plan}, \pi^\mathrm{track})}$ that consists of 1) a trajectory planner
\begin{equation*}
    \pi^\mathrm{plan}: \R^{d_x} \to \R^{Td_x}
\end{equation*}
that takes in an initial condition $\xi \in \R^{d_x}$ and outputs a reference trajectory $r_{0:T} \in \mathcal{X}$, and 2) a tracking controller
\begin{equation*}
    \quad \pi^\mathrm{track}: \R^{d_x} \times \R^{Td_x} \to \R^{d_u}
\end{equation*}
that takes in the current state and a reference trajectory to output a control action to best track the given trajectory. We now decompose problem \eqref{eq:original-prob} such that it may inform a suitable parameterization for the planning and tracking policies, $\pi^\mathrm{plan}$ and $\pi^\mathrm{track}$.

\section{Layered Approach to Optimal Control}\label{sec:layering}
We first consider a variation of problem~\eqref{eq:original-prob} with a fixed initial condition $\xi$, and rewrite it into a form that has a natural layered control architecture interpretation. For ease of notation, we use unsubscripted letters $x, u, r$ to denote the respective trajectories stacked as a column vector
\begin{equation*}
    x := x_{0:T}, \quad u := u_{0:T-1}, \quad r := r_{0:T}.
\end{equation*}
 We begin the rewrite of problem~\eqref{eq:original-prob} by introducing a redundant variable $r = x$ to get an equivalent problem
\begin{equation}\label{eq:redundant-prob}
    \begin{aligned}
        \underset{r, x, u}{\text{minimize}}&\quad
            \mathcal{C}(r) + \mathcal{D}(u)\\
        \text{subject to}&\quad x_{t+1} = f(x_t, u_t), \quad \forall t=0, 1, ..., T-1,\\
        &\quad r\in \mathcal{X}, \quad u \in \mathcal{U}, \quad x_0 = \xi, \quad r = x,
    \end{aligned}
\end{equation}
where we use the fact that $r=x$ to move the state cost and constraint from $x$ onto $r$. Defining the indicator functions
\begin{gather*}
    \mathbb{I}_{dyn}(x, u) = \begin{cases}
        0, & \begin{aligned}
            &x_0 = \xi,\\ &x_{t+1} = f(x_t, u_t), u \in \mathcal{U},
        \end{aligned}\\
        \infty, & \text{otherwise}
    \end{cases},\\
    \mathbb{I}_{state}(r) = \begin{cases}
        0, & r \in \mathcal{X}\\
        \infty, & \text{otherwise}
    \end{cases},
\end{gather*}
we write the partial augmented Lagrangian of problem~\eqref{eq:redundant-prob} in terms of the (scaled) dual variable $\nu$
\begin{equation}
    \begin{aligned}
        \mathcal{L}_{\rho}(r, x, u, \nu)
        %
        &=\;\mathcal{C}(r) + \mathcal{D}(u) + \mathbb{I}_{dyn}(x, u) + \mathbb{I}_{s}(r) + \frac{\rho}{2} \lVert r + \nu - x \rVert_2^2 - \frac{\rho}{2} \lVert \nu \rVert_2^2.
    \end{aligned}
\end{equation}
\noindent Applying dual ascent to this augmented Lagrangian, we obtain the following method-of-multiplier updates
\begin{align}
    &
    \begin{aligned}
        \quad (r^+, x^+, u^+) =\underset{x, u, r}{\text{argmin}}\;&\mathcal{C}(r) + \mathcal{D}(u) + \frac{\rho}{2} \lVert r + \nu - x \rVert_2^2\\
        \text{s.t.}\;& x_{t+1} = f(x_t, u_t),\quad\forall t,\\
        &r\in \mathcal{X},\; u\in\mathcal{U},\; x_0 = \xi
    \end{aligned}\label{eq:update-1}\\
    &\quad \nu^+ = \nu + (r^+ - x^+),\label{eq:update-2}
\end{align}
which will converge to locally optimal primal and dual variables $r^*, x^*, u^*, \nu^*$ given mild assumptions on the smoothness and convexity of $\mathcal{C}, \mathcal{D}$ and the constraints in the neighborhood of the optimal point (See \citet[\S 2]{bertsekas2014constrained}).

For a layered interpretation, we note that the primal update \eqref{eq:update-1} can be written as a nested optimization problem
\begin{equation}\label{eq:upper-level-opt}
    \begin{aligned}
        r^+ = \underset{r}{\text{minimize}}&\quad\mathcal{C}(r)+ p^*(r + \nu; \xi)\\
        \text{s.t.}&\quad r\in \mathcal{X}
    \end{aligned}
\end{equation}
where $p^*(r +\nu; \xi)$ is the locally optimal value of the $(x, u)$-minimization step
\begin{equation}\label{eq:lower-level-opt}
    \begin{aligned}
        p^*(r+\nu; \xi) = \underset{x, u}{\text{min}}&\quad\mathcal{D}(u) + \frac{\rho}{2} \lVert r + \nu - x \rVert_2^2\\
        \text{s.t.}&\quad x_{t+1} = f(x_t, u_t),\; u \in \mathcal{U},\; \forall t,\\
        &\quad x_0 = \xi.
    \end{aligned}
\end{equation}
We immediately recognize that optimal control problem~\eqref{eq:lower-level-opt} is finding the control action $u$ to minimize a quadratic tracking cost for the reference trajectory
\[\tilde{r}:=r + \nu.\]
Thus, this nested rewrite can be seen as breaking the primal minimization problem \eqref{eq:update-1} into a trajectory optimization problem \eqref{eq:upper-level-opt} that seeks to find the best reference $r$ and a tracking problem \eqref{eq:lower-level-opt} that seeks to best track the \textit{perturbed} trajectory $\tilde{r}$. A subtlety here is that the planned trajectory, $r$, and the trajectory sent to the tracking controller, $\tilde{r}$, are different. To understand this discrepancy, let us first consider a similar, but perhaps more intuitive, reference optimization problem:
\begin{equation}\label{eq:heuristic-upper-level-opt}
    \begin{aligned}
        \underset{r}{\text{minimize}}&\quad\mathcal{C}(r)+ p^*(r; \xi)\\
        \text{s.t.}&\quad r\in \mathcal{X}.
    \end{aligned}
\end{equation}
This heuristics-based approach, employed in previous works such as \citet{srikanthan2023data, zhang2024change}, seeks to find a reference that balances minimizing the nominal cost $\mathcal{C}(r)$ and not incurring high tracking cost $p^*(r; \xi)$.
In these works, the solution $r$ is then sent to the tracking controller unperturbed.

A problem with this approach is that unless the tracking controller can execute the given reference perfectly, the executed trajectory $x$ will differ from the planned reference $r$. One can mitigate this deviation by multiplying the tracking cost with a large weight, but this can quickly become numerically ill-conditioned, or bias the planned trajectory towards overly conservative and easy-to-track behaviors.
In these works, the solution $r$ is then sent to the tracking controller unperturbed.
A problem with this approach is that unless the tracking controller can execute the given reference perfectly, the executed trajectory $x$ will differ from the planned reference $r$. One can mitigate this deviation by multiplying the tracking cost with a large weight, but this can quickly become numerically ill-conditioned, or bias the planned trajectory towards overly conservative and easy-to-track behaviors.

Returning to the method-of-multiplier updates~\eqref{eq:update-1} and~\eqref{eq:update-2}, we note that, under suitable technical conditions, solving the planning layer problem~\eqref{eq:upper-level-opt} using the locally optimal dual variable $\nu^*$ leads to the feasible solution satisfying $r^*=x^*$.  In particular, the perturbed reference trajectory $\tilde{r}^\star=r^\star+\nu^*$ is sent to the tracking controller defined by problem~\eqref{eq:lower-level-opt}, and this results in the executed state trajectory $x^*$ matching the reference $x^* = r^*$.  This discussion highlights the role of the locally optimal dual variable as coordinating the planning and tracking layers, and motivates our approach of explicitly modeling this dual variable in our learning framework.


Following this intuition, in the next section, we show how to parameterize $\pi^\mathrm{plan}$ and $\pi^\mathrm{track}$ to approximately solve \eqref{eq:upper-level-opt} and \eqref{eq:lower-level-opt}, respectively. In practice, finding $\nu^*$ with the iterative update in \eqref{eq:update-2} can be prohibitively expensive. To circumvent this issue, we recognize that any locally optimal dual variable $\nu^*$ can be written as a function of the initial condition $\xi$. We thus seek to learn an approximate map to predict this locally optimal dual variable $\nu^*$ from the initial condition $\xi$.\footnote{We have been somewhat cavalier in our assumption that such a locally optimal dual variable $\nu^*$ exists. We note that notions of local duality theory, see for example~\citep[Ch 14.2]{luenberger1984linear}, guarantee the existence of such a locally optimal dual variable under mild assumptions of local convexity.}


We close this section by noting that the above derivation assumes that the reference trajectory is of the same dimension as the state, i.e., that $r_t = x_t$. However, if the state cost $\mathcal{C}$ and constraints $\mathcal{X}$ only require a subset of the states, i.e., if they are defined in terms of $z_t = g(x_t) \in \R^{d_z}$, with $d_z < d_x$, then one can modify the discussion above by replacing the redundant constraint $x = r$ with $z = r$, so that the reference only needs to be specified on the lower dimensional output $z$. We refer the readers to Appendix~\ref{sec:subset-plan} for the details.
\section{Actor-Critic Learning in the Layered Control Architecture}\label{sec:method}

\subsection{Parameterization of the Layered Policy}
We parameterize our layered policy $\pi=(\pi^\mathrm{plan}, \pi^\mathrm{track})$ so that its structure parallels the dual ascent updates \eqref{eq:upper-level-opt} and \eqref{eq:lower-level-opt}. The tracking controller $\pi^\mathrm{track}_\phi: \R^{d_x} \times \R^{Td_x} \to \R^{d_u}$, specified by learnable parameters $\phi$, seeks to approximate a feedback controller that solves the tracking problem \eqref{eq:lower-level-opt}.\footnote{The finite-horizon nature of \eqref{eq:lower-level-opt} calls for a time-varying controller. Thus, the correct $\pi^\mathrm{track}$ and associated value function $p^\pi$ need to be conditioned on the time step $t$. In our experiments, we show that approximating this with a time-invariant controller works well for the time horizons we consider.} The trajectory generator $\pi^\mathrm{plan}_{\theta,\psi}$ seeks to approximately solve the planning problem~\eqref{eq:upper-level-opt}.  It has learnable parameters $\theta$ and $\psi$ and is defined as the solution to the optimization problem
\begin{equation}\label{eq:traj-gen-param}
    \begin{aligned}
        \pi^\mathrm{plan}_{\theta, \psi}(\xi) = \underset{r}{\text{minimize}}&\quad\mathcal{C}(r)+ p_{\psi}^{\pi^\mathrm{track}}(r + v_{\theta}(\xi); \xi)\\
        \text{s.t.}&\quad r\in \mathcal{X}.
    \end{aligned}
\end{equation}
Thus $\pi^\mathrm{plan}$ generates a reference trajectory from initial condition $\xi$ by solving problem~\eqref{eq:traj-gen-param}. The objective of this optimization problem contains two learned components, $v_\theta$ and $p_{\psi}^{\pi^\mathrm{track}}$, specified by parameters $\theta$ and $\psi$, respectively. First, $v_\theta: \R^{d_x} \to \R^{Td_x}$ is a dual network that seeks to predict the locally optimal dual variable $\nu^*$ from initial condition $\xi$. Then, the tracking value function $p_{\psi}^{\pi^\mathrm{track}}: \R^{Td_x} \times \R^{d_x} \to \R$ takes in an initial state $\xi$ and a reference trajectory $r$ and learns to predict the quadratic tracking cost \eqref{eq:lower-level-opt} that the policy $\pi^\mathrm{track}$ will incur on this reference trajectory. Summarizing, our layered policy consists of three learned components: the dual network $v_\theta$, the low-layer tracking policy $\pi^\mathrm{track}_\phi$, and its associated value function $p_{\psi}^{\pi^\mathrm{track}}$. In what follows, we explain how we learn the tracking value function $p_{\psi}^{\pi^\mathrm{track}}$ and policy $\pi^\mathrm{track}_\phi$ jointly via the actor-critic method, and how to update the dual network $v_\theta$ in a way similar to dual ascent.

\subsection{Learning the Tracking Controller via Actor-Critic Method} \label{sec:learning-tracking-algorithm}
We use the actor-critic method to jointly learn the tracking value function $p_{\psi}^{\pi^\mathrm{track}}$ and policy $\pi^\mathrm{track}_\phi$. We are learning a deterministic policy and its value function, a setting that has been extensively explored and for which many off-the-shelf algorithms exist \citep{silver2014deterministic, lillicrap2015continuous, fujimoto2018addressing}. In what follows, we specify the RL problem for learning the tracking controller and treat the actor-critic algorithm as a black-box solver for finding our desired parameters $\phi$ and $\psi$.

We define an augmented system with the state $x^{aug}_t = (x_t, \mathbf{r}_t)^\top \in \R^{(H+1)d_x}$, which concatenates $x_t$ with a $H$-step reference trajectory $\mathbf{r}_t=(r_{t}^\top\; r_{t+1}^\top\; \cdots\; r_{t+H-1}^\top)^\top$, where $H \in \mathbb{Z}^+$ specifies the tracking controller's horizon of look-ahead. The augmented state transitions are then given by
\begin{equation}\label{eq:augmented-state}
    x^{aug}_{t+1} 
    = \begin{bmatrix} f(x_t, u_t) \\ \mathcal{Z}\mathbf{r}_t \end{bmatrix}, t=1,...T,
\end{equation}
where $\mathcal{Z}$ is a block-upshift operator that shifts the reference trajectory forward by one timestep. The cost of the augmented system $c^{aug}$ is chosen to match the tracking optimization problem \eqref{eq:lower-level-opt}, i.e., we set
\begin{equation}\label{eq:tracking-stage-cost}
    c^{aug}(x^{aug}_t, u_t) = \frac{\rho}{2}\norm{x_{t+1} - r_{t+1}}_2^2 + \mathcal{D}_t(u_t).
\end{equation}
The initial condition $x^{aug}_0$ is found by first sampling $\xi \sim D_\xi$, and then setting $\mathbf{r}_0$ to the first $H$ steps of the reference generated by $\pi^\mathrm{plan}(\xi)$. We then run the actor-critic algorithm on this augmented system to jointly learn $p_{\psi}^{\pi^\mathrm{track}}$ and $\pi^\mathrm{track}_\phi$.

\subsection{Learning the Dual Network} \label{sec:learning-dual-algorithm}
We design our dual network update as an iterative procedure that mirrors the dual ascent update step~\eqref{eq:update-2}, which moves the dual variable in the direction of the mismatch between reference $r^+$ and execution $x^+$. At each iteration, we sample a batch of initial conditions $\{\xi_i\}_{i=1}^B$, and for each $\xi_i$, we solve the planning problem \eqref{eq:traj-gen-param} with current parameters $\phi$ and $\theta$ to obtain reference trajectories $r_i = \pi^\mathrm{plan}_{\theta^{(k)}, \psi}(\xi).$  We then send the perturbed trajectories $\tilde{r}_i = r_i + v_\theta(\xi_i)$ to the tracking controller to obtain the executed trajectories
\begin{equation*}
    x_{i, t+1} = f(x_{i,t}, \pi^\mathrm{track}_\phi(x_{i,t}, \tilde{\mathbf{r}}_{i,t})),\quad t=0, ...,T.
\end{equation*}
Similar to the dual ascent step, we then perform a gradient ascent step in $\theta$ to move $v_{\theta^{(k)}}(\xi_i)$ in the direction of $r_i-x_i$:
\begin{equation}\label{eq:dual-update}
\begin{aligned}
    \theta^{+} &\leftarrow \theta + \eta \left( \nabla_\theta \sum_{i=1}^{B} \frac{1}{B} \left(r_i - x_i\right)^\top v_\theta(\xi_i) \right)\\
    &= \theta + \eta \sum_{i=1}^{B} \frac{1}{B} \left(r_i - x_i\right)^\top J_{v,\theta}(\xi_i; \theta),
\end{aligned}
\end{equation}
where $J_{v, \theta}$ denotes the Jacobian of $v$ w.r.t. $\theta$. Note that even though $r_i$ and $x_i$ implicitly depend on $\theta$, similar to the dual ascent step \eqref{eq:update-2}, we do not differentiate through these two terms when computing this gradient. In the next section, we show that for the case of linear quadratic regulators, this update for the dual network parameter $\theta$ converges to the vicinity of the optimal parameter $\theta^*$ if the tracking problem is solved to sufficient accuracy.

\subsection{Summary of the Algorithm}
We summarize our algorithm in Algorithm \ref{alg:bi-actor-critic}. The outer loop of the algorithm (Line 1-9) corresponds to the dual update procedure described in Section \ref{sec:learning-dual-algorithm}. Within each iteration of the outer loop, we also run the actor-critic algorithm to update the tracking policy $\pi_{\phi^{(k)}}^\mathrm{track}$ and its value function $p^\pi_\psi$ (Line 5-8). Note that we do not wait for the tracking controller to converge before starting the dual update. In Section~\ref{sec:experiment}, we empirically validate that dual learning can start to make progress even when the tracking controller is still suboptimal. After the components are learned for the specified iterations, we directly apply the learned policy $\pi^\mathrm{plan}, \pi^\mathrm{track}$ for any new initial condition $\xi$.

\begin{algorithm}[h]
    \caption{Layered Actor-Critic}\label{alg:bi-actor-critic}
    \KwResult{Policy parameters $\phi, \psi, \theta$}
    
    \For{$k = 1, ..., K$}{
        Sample a batch of initial conditions $\{\xi^{(k)}_i\}_{i=1}^{B}$\;
        Predict the optimal dual variables $\hat{\nu}^{(k)}_i = v_{\theta^{(k)}}(\xi_i^{(k)})$\;
        Solve \eqref{eq:traj-gen-param} to find reference trajectories $\{r_i^{(k)}\}_{i=1}^{B}$\;
        Construct augmented state $x^{aug}_{i,0} = [\xi_i, r_i^{(k)}]^\top$\;
        \For{$t=0, ..., T-1$}{
            Roll augmented dynamics forward with $\pi_{\phi^{(k)}}^\mathrm{track}$ to get $\{x_{i,t+1}^{(k)}\}_{i=1}^{B}$\;
            Update $\pi_{\phi}^\mathrm{track}$ and $p^\pi_\psi$ with observed transition using actor-critic algorithm\;
        }
        Update the dual network parameter per \eqref{eq:dual-update}\;
    }
\end{algorithm}
\section{Analysis for Linear Quadratic Regulator}\label{sec:analysis}
In this section, we consider the unconstrained linear quadratic regulator (LQR) problem and show that our method learns to predict the optimal dual variable if we solve the tracking problem well enough. We focus on the dual update because the tracking problem \eqref{eq:lower-level-opt} reduces to standard LQR, to which existing results \citep{bradtke1994adaptive, tu2018least} are readily applicable. In what follows, we define the problem we analyze, and first show that dual network updates of the form \eqref{eq:dual-update} converge to the optimal dual map if one perfectly solves the planning \eqref{eq:upper-level-opt} and tracking problem \eqref{eq:lower-level-opt}. We then present a robustness result which shows that the algorithm will converge to the vicinity of the optimal dual variable if we solve the tracking problem with a small error.

We consider the instantiation of \eqref{eq:redundant-prob} with the dynamics
\begin{equation}\label{eq:linear-dynamics}
    x_{t+1} = f(x_t, u_t) = A x_t + B u_t
\end{equation}
\vspace{-2.5mm}
and cost functions
\begin{equation}\label{eq:quadratic-cost}
\begin{gathered}
    \mathcal{C}(r) = \sum_{t=0}^{T} r_t^\top Q r_t =: r^\top \mathcal{Q} r,\\
    \mathcal{D}(u) = \sum_{t=0}^{T-1} u_t^\top R u_t =: u^\top \mathcal{R} u,
\end{gathered}
\end{equation}
where $Q \succeq 0, R \succ 0$, $\mathcal{Q} = I_{T} \otimes Q$ and $\mathcal{R} = I_{T-1} \otimes R$. States and control inputs are unconstrained, i.e., $\mathcal{X} = \R^{d_x}, \mathcal{U} = \R^{d_u}$. The initial condition $\xi$ is sampled i.i.d. from the standard normal distribution $\mathcal{N}(0, I)$.

In this case, strong duality holds, and the optimal dual variable\footnote{If not further specified, when we refer to $\nu$ or the dual variable, we mean the dual variable associated with the constraint $r=x$ in problem \eqref{eq:redundant-prob}} $\nu^*$ is a linear function of the initial condition $\xi$. (See Lemma \ref{lem:dual-existence} in Appendix \ref{sec:appendix-proof-opt}.) We thus parameterize the dual network as a linear map
\begin{equation}\label{eq:lqr-dual-pred}
    v_\theta(\xi) = \Theta \xi.
\end{equation}

\subsection{With Optimal Tracking}
We first consider the following update rule, wherein we assume that the planning \eqref{eq:upper-level-opt} and tracking problems \eqref{eq:lower-level-opt} are solved optimally. At each iteration, we first sample a minibatch of initial conditions $\{\xi^{(k)}_i\}_{i=1}^{B}, \, \xi_i^{(k)}\overset{i.i.d.}{\sim} \mathcal{N}(0, I)$, and use the current $\Theta^{(k)}$ to predict the optimal dual variable
\begin{equation*}
    \hat{v}^{(k)}_i = v_{\theta^{(k)}}(\xi_i) = \Theta^{(k)} \xi_i.
\end{equation*}
We assume we perfectly solve the trajectory optimization problem
\begin{equation}\label{eq:opt-traj-opt-update}
    r^{(k)}_i = \underset{r}{\text{argmin}}\ r^\top \mathcal{Q} r + p^*(r + \hat{v}_i^{(k)}; \xi_i),
\end{equation}
where $p^*(\cdot)$ is the optimal value of the tracking problem
\begin{equation}\label{eq:opt-tracking-update}
\begin{aligned}
        x^{(k)}_i, u^{(k)}_i = \underset{x, u}{\text{argmin}}&\ u^\top \mathcal{R} u + \frac{\rho}{2} \lVert r^{(k)}_i + \hat{v}^{(k)}_i - x\rVert_2^2\\
        \text{s.t.}\quad& x_{t+1} = Ax_t + Bu_t, \quad x_0 = \xi.
\end{aligned}
\end{equation}
This is a standard LQR optimal control problem, and closed-form expressions for the optimizers and the value function are readily expressed in terms of the solution to a discrete algebraic Riccati equation.

After solving \eqref{eq:opt-tracking-update}, we update the dual map $\Theta$ as
\begin{equation}\label{eq:lqr-dual-update}
    \begin{aligned}    
    \Theta^{(k+1)} &= \Theta^{(k)} + \eta \nabla_\Theta \left( \sum_{i=1}^{B} \frac{1}{B} \left(r^{(k)}_i - x^{(k)}_i\right)^\top v_{\theta^{(k)}}(\xi_i)\right)\\
    &=  \Theta^{(k)} + \eta\sum_{i=1}^{B} \frac{1}{B} \left(r^{(k)}_i - x^{(k)}_i\right) \xi_i^\top
    \end{aligned}
\end{equation}

A feature of this update rule is that the difference between the reference $r_i$ and the executed trajectory $x_i$ can be written out in closed form as follows.

\begin{lemma}\label{lem:rx-expression-simple}
Given the update rules \eqref{eq:opt-traj-opt-update}, \eqref{eq:opt-tracking-update}, the difference between the updates $r^{(k)}_i$ and $x^{(k)}_i$ can be written as a linear map of the initial condition $\xi$ as
\begin{equation*}
    r^{(k)}_i  - x^{(k)}_i  = H \Theta^{(k)} \xi_i + G \xi_i,
\end{equation*}
where $H$ and $G$ are matrices of appropriate dimensions that depend on $A, B, Q, R$, and $H$ is symmetric negative definite.  See Lemma~\ref{lem:rx-expression} in Appendix~\ref{sec:appendix-proof-opt} for definitions of $H$ and $G$.
\end{lemma}

We leverage Lemma~\ref{lem:rx-expression-simple}, and that the matrix $H$ is negative definite, to show that the updates~\eqref{eq:lqr-dual-pred}-\eqref{eq:lqr-dual-update} make progress in expectation.
\begin{theorem}\label{thm:opt-tracking}
    Consider the cost functions \eqref{eq:quadratic-cost} and dynamics \eqref{eq:linear-dynamics}, and fix an initial $\Theta^{(0)}$. 
    Fix a step size ${\eta = \frac{2}{\smax(H) - \smin(H)}}$ and mini-batch size ${B > \frac{2d_x\smax^2(H)}{\smin^2(H)}}$.   The iterates generated by the  updates~\eqref{eq:lqr-dual-pred}-\eqref{eq:lqr-dual-update} satisfy
    \begin{equation*}
        \mathbb{E} \norm{\Theta^{(k)} - \Theta^{*}}_2 \leq \gamma^k \mathbb{E} \norm{\Theta^{(0)} - \Theta^{*}}_2,
    \end{equation*}
    where $\gamma\in(0,1)$ is a function of $\eta$, $H$, $B$, and $d_x$.
\end{theorem}
\begin{proof}
    See Appendix \ref{sec:appendix-proof-opt}.
\end{proof}
\vspace{-2mm}
\subsection{With Suboptimal Tracking}\label{sec:analysis-subopt}
We consider the case where we only have approximate solutions to the updates \eqref{eq:opt-tracking-update} and \eqref{eq:opt-traj-opt-update}. We leverage the structural properties of the LQR problem, and parameterize the optimal tracking controller as a linear map, and its value function as a quadratic function of the augmented state. Denote $F\xi$ as the open-loop response of initial condition $\xi$, we consider perturbations in the optimal value function $p^*$ as
\begin{equation}\label{eq:r-update-perturb}
    \hat{p}(\tilde{r}, \xi) = p^*(\tilde{r}; \xi) + (\tilde{r} - F\xi) \Delta_P (\tilde{r} - F\xi),
\end{equation}
and perturbations in the control action as
\begin{equation}\label{eq:u-update-perturb}
    \hat{u}(\tilde{r}, \xi) = u^*(\tilde{r}, \xi) + \Delta_{u,r}\tilde{r} + \Delta_{u, \xi} \xi
\end{equation}
where $u^*$ denotes the $u$ solution of \eqref{eq:opt-tracking-update}. We note that the perturbations $\Delta_P, \Delta_{u,r}, \Delta_{u, x_0}$ represent the difference between learned and optimal policies, and have been shown to decay with the number of transitions used for training \citep{bradtke1994adaptive, tu2018least}. Perturbation analysis on Theorem \ref{thm:opt-tracking} shows that if the learned controller is close to optimal, the dual map $\Theta$ will converge to a small ball around $\Theta^*$, where the radius of the ball depends on the error of the learned tracking controller.  Due to space constraints, we present an informal version of this result here, and relegate a precise statement and proof to Appendix~\ref{sec:appendix-proof-subopt}.
\begin{theorem} \label{thm:subopt-tracking}(informal)
Consider the dynamics \eqref{eq:linear-dynamics} and cost \eqref{eq:quadratic-cost}. Consider the update rules \eqref{eq:lqr-dual-pred}-\eqref{eq:lqr-dual-update} with the perturbations \eqref{eq:r-update-perturb} and \eqref{eq:u-update-perturb}. Denote the size of the perturbations as $\epsilon_P = \norm{\Delta_P}, \epsilon_{u,r} = \norm{\Delta_{u,r}}, \epsilon_{u,\xi} = \norm{\Delta_{u,\xi}}$. Given any $\Theta^{(0)}$, if the perturbations $\epsilon_P, \epsilon_{u,r}$ are sufficiently small, there exist step size $\eta$ and batch size $B$ such that
\vspace{-1mm}
\begin{equation*}
    \E \norm{\Theta^{(k)} - \Theta^*} \leq \gamma^k \E \norm{\Theta^{(0)} - \Theta^*} + \frac{1-\gamma^k}{1-\gamma} e(\epsilon_P, \epsilon_{u,r}, \epsilon_{u,\xi}),
\end{equation*}
where $0 < \gamma < 1$, $e(\epsilon_P, \epsilon_{u,r}, \epsilon_{u,\xi})$ is an error term depending polynomially on its arguments.
\end{theorem}

\section{Experiments}\label{sec:experiment}
We now proceed to evaluate our algorithm numerically on LQR and unicycle systems. For all the experiments, we use the CleanRL \citep{huang2022cleanrl} implementation of Twin-Delayed Deep Deterministic Policy Gradient (TD3) \citep{fujimoto2018addressing} as our actor-critic algorithm. All code needed to reproduce the examples found in this section will be made available at the following repository: \url{https://github.com/unstable-zeros/layered-ac}.
\vspace{-2mm}
\subsection{Unconstrained LQR}
\paragraph{Experiment Setup} We begin by validating our algorithm on unconstrained LQR problems and show that our algorithm achieves near-optimal performance and near-perfect reference tracking. We consider linear systems \eqref{eq:linear-dynamics} with dimensions ${d_x=d_u=2, 4, 6, 8}$ and horizon $T=20$. For each system size, we randomly sample 15 pairs of dynamics matrices $(A, B)$\footnote{Each entry is sampled i.i.d from the standard normal distribution.} and normalize $A$ so that the system is marginally stable ($\rho(A)=1$). For all setups, we consider a quadratic cost \eqref{eq:quadratic-cost} with $Q = I_{d_x}, R = 0.01 I_{d_u}$. We have $\mathcal{X}=\R^{d_x}, \mathcal{U}=\R^{d_u}$, and the initial state $\xi \sim \N(0, I_{d_x})$. We leverage the linearity of the dynamics to parameterize the tracking controller $\pi^{\mathrm{track}}$ to be linear, and the value function $p^\pi$ to be quadratic in the augmented state \eqref{eq:augmented-state}. Since $p^\pi$ is quadratic, the optimization problem for the trajectory planner \eqref{eq:traj-gen-param} is a QP, which we solve with CVXPY \citep{diamond2016cvxpy}. We parameterize the dual network to be a linear map as in \eqref{eq:lqr-dual-pred}. We train the tracking policy and the dual network jointly for $100,000$ transitions ($5,000$ episodes) with dual batch size $B=5$, before freezing the tracking policy and just updating the dual network for another $5,000$ transitions (250 episodes). We specify the detailed training parameters in Table \ref{tab:lqr-hyperparam} in the Appendix. During training, we periodically evaluate the learned policy by applying it on $50$ initial conditions. We then record the cost it achieved and the average tracking deviation $\frac{1}{T}\sum_{t=1}^T|r_t - x_t|$. We report relative costs normalized by the optimal cost of solving \eqref{eq:redundant-prob} directly with the corresponding true dynamics and cost function. Thus, a relative cost of $1$ is optimal. The results are summarized below.

\begin{table}[h]
    \centering
    \begin{tabular}{|c|c|c|}
         \hline$d_x, d_u$ & Relative Cost ($\downarrow$) & Mean Tracking Deviation ($\downarrow$) \\ \hline
         2              & 1.004      & 0.002 \\\hline
         4              & 1.009      & 0.003 \\\hline
         6              & 1.020      & 0.008 \\\hline
         8              & 1.031      & 0.009 \\\hline
    \end{tabular}
    \caption{LQR Results on Varying System Sizes.}
    \label{tab:lqr-result}
\end{table}

\paragraph{Varying System Sizes} In Table \ref{tab:lqr-result}, we summarize the cost and mean tracking deviations evaluated at the end of training.\footnote{The reported numbers are their respective medians taken over $15$ random LQR instances.} We first note that the learned policy achieves near-optimal cost and near-perfect tracking for all the system sizes considered. Figure \ref{fig:clqr-traj} shows a representative sample trajectory that has a mean tracking deviation of $0.005$. This shows that our parameterization and learning algorithm are able to find good policies with only black-box access to the underlying dynamics.  We note that the performance degrades slightly as the size of the system grows. This is likely because learning the tracking controller becomes more difficult as the size of the state space increases. However, even for the largest system we considered ($d_x=8$), the cost of the learned controller is still only {$3\%$} above optimal.

\begin{figure}
    \centering
    \includegraphics[width=0.8\textwidth, trim={0.5cm 0 0.7cm 1cm},clip]{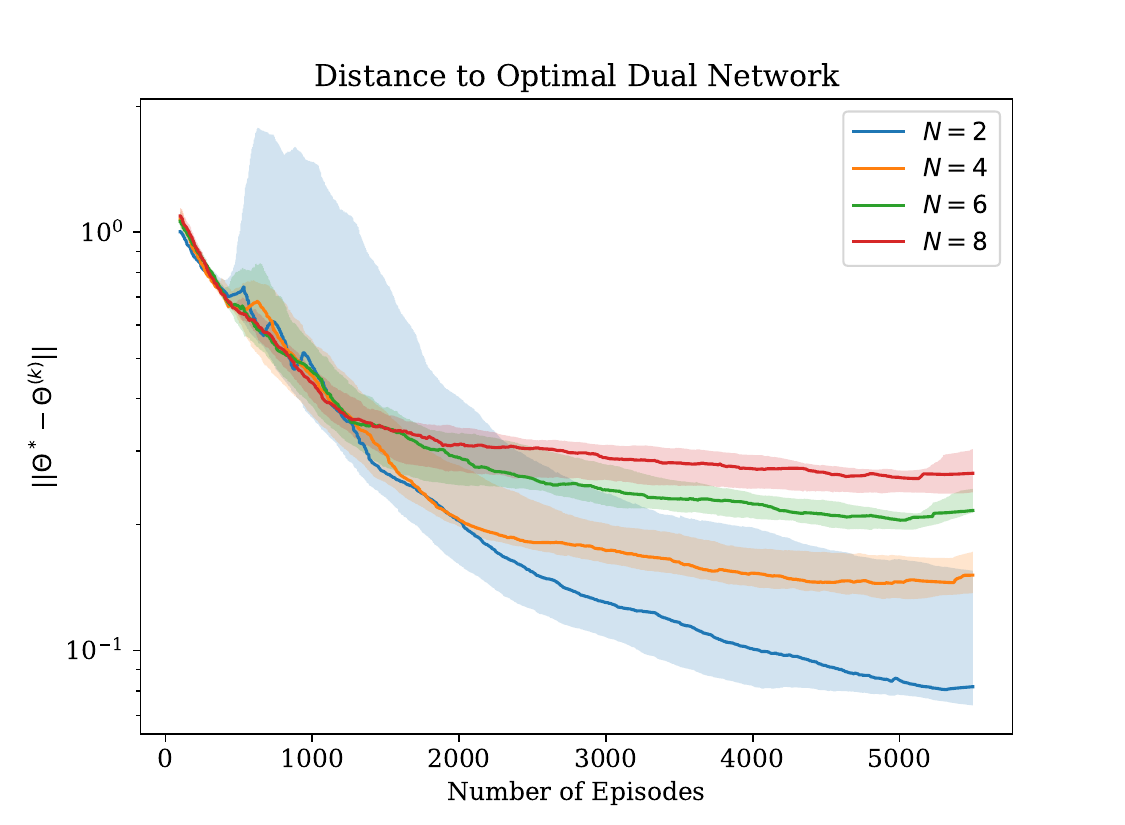}
    \caption{Training progress for the dual map parameter $\Theta$. Here, the solid lines are the median over $15$ random LQR instances, and the shaded regions represent the $25^{th}$ to $75^{th}$ percentile.}
    \label{fig:theta-distance}
\end{figure}

\paragraph{Visualization of Dual Learning} We visualize the algorithm's progress for learning the dual map in Figure \ref{fig:theta-distance}. Recall that our theory suggests that in the unconstrained LQR case, the dual map weight $\Theta$ will converge to the neighborhood of the optimal dual map $\Theta^*$, where the radius of the neighborhood depends on the quality of the learned controller. This is indeed the case shown in Figure \ref{fig:theta-distance}, where the norm of the difference $\Theta - \Theta^*$ first decays exponentially before reaching a plateau. We note that this plot also validates our choice to start learning the dual network before the tracking controller training has converged, as progress is made starting at the very beginning of the training.

\begin{table}[h]
    \centering
    \begin{tabular}{|c|c|c|}
         \hline$d_x, d_u$ & Relative Cost ($\downarrow$) & Mean Tracking Deviation ($\downarrow$) \\ \hline
         2    & 1.012 ($+0.7\%$)     & 0.046 ($+2,300\%$) \\ \hline
         4    & 1.028 ($+1.8\%$)     & 0.045 ($+1,500\%$)\\ \hline
         6    & 1.036 ($+1.5\%$)     & 0.061 ($+763\%$)\\ \hline
         8    & 1.052 ($+2.0\%$)     & 0.062 ($+689\%$)\\ \hline
    \end{tabular}
    \caption{LQR Results without Dual Learning. Numbers in parentheses denote the percentage difference from the approach with dual learning.}
    \label{tab:lqr-result-no-dual}
    \vspace{-5mm}
\end{table}
\paragraph{Comparison to heuristic approach} We now compare our approach to the heuristic approach of generating trajectories without using the learned dual variable \citep{srikanthan2023data, zhang2024change}, summarized in equation \eqref{eq:heuristic-upper-level-opt}. We use the same parameters to train a tracking controller and a value function, with the only difference being that $\pi^\mathrm{plan}$ solves \eqref{eq:heuristic-upper-level-opt} instead of \eqref{eq:traj-gen-param}. We show the results in Table \ref{tab:lqr-result-no-dual}.
First, the heuristic policy is outperformed by our approach both in terms of cost and tracking deviation across all the different system sizes, showing the value of learning to predict the dual variable. We note that the difference is especially pronounced for tracking deviation. Since the dual network learned to preemptively perturb the reference to minimize tracking error, it achieves near-perfect tracking and an order of magnitude lower tracking error. This suggests that learning the dual network is especially important in achieving good coordination between the trajectory planner and the tracking controller.

\begin{table}[h]
    \centering
    \begin{tabular}{|c|c|c|c|c|c|} \hline
         $\rho$            &  0.5 & 1 & 2 & 4 & 8    \\ \hline
         Relative Cost ($\downarrow$)  & 2.04 & 1.24 & 1.11& 1.10 & 1.19\\ \hline
         Mean Deviation ($\downarrow$)  & 0.039 & 0.01 & 0.005 & 0.003 & 0.003\\ \hline
    \end{tabular}
    \caption{LQR Results on Varying Hyperparameter $\rho$}
    \label{tab:rho-result}
\end{table}

\paragraph{The role of $\rho$} Finally, we note that the penalty parameter $\rho$ is a hyperparameter that needs to be tuned when implementing Algorithm \ref{alg:bi-actor-critic}. Since $\rho$ directly affects the objective of the tracking problem, it begs the question of whether the choice of $\rho$ significantly affects the performance of our algorithm. We test this hypothesis on 15 randomly sampled \textit{underactuated} systems where $d_x=4$ and $d_u=2$. We use the same set of hyperparameters as above except for $\rho$. We report the results in Table \ref{tab:rho-result}. From Table \ref{tab:rho-result}, we see that algorithm behavior is robust to the choice of $\rho$, so long as it is large enough; indeed, only the case of $\rho=0.5$ leads to significant performance degradation. 

\subsection{LQR with State Constraints}
In the unconstrained case, the map from the initial condition $\xi$ to the optimal dual variable $\nu^*$ is linear. In this section, we consider the case where inequality constraints are introduced and this map is no longer linear. We show that by parameterizing the dual map $v_\theta(\xi)$ as a neural network, we can learn well-performing policies that respect the constraints. Similar to the experiments above, we randomly sample $10$ LQR systems where $d_x = d_u = 2$. Here we consider stable systems with $\rho(A)=0.995$. The time horizon is fixed to $T=20$ and cost matrices are $Q = I, R = 0.01 I$. We add the constraint that
\begin{equation*}
    \mathcal{X} = \{x_{0:T}\;|\;x_{t,i} \geq -0.05,\quad 1 \leq t \leq 20, i = 1, 2\},
\end{equation*}
i.e., that we restrict all states except for the initial state to be above $-0.05$. Since the additional constraint does not affect the tracking problem, we still parametrize the actor and critic as linear and quadratic, respectively. Since the optimal dual map is no longer linear, we parameterize the dual map as a neural network with a single hidden layer with ReLU activation. Note that the optimization problem for trajectory planning \eqref{eq:traj-gen-param} is still a QP as it does not depend on the form of the dual network. To account for the nonlinearity of the dual network, we increase the dual batch size to $40$ trajectories, and train the policy and dual network for $150,000$ transitions, before freezing the tracking controller and training the dual network for another $600,000$ transitions ($30,000$ episodes). We specify the detailed training parameters in Table \ref{tab:clqr-hyperparam}. We report the relative cost and mean constraint violation\footnote{We measure the constraint violation as $max(0.05 - x^{(i)}_t, 0), t = 1, ..., 20$. Reported values are the medians over the $10$ systems.} in Table \ref{tab:clqr-result} and show a representative sample trajectory in Figure \ref{fig:clqr-traj}.

\begin{table}[]
    \centering
    \begin{tabular}{|c|c|c|}
        \hline Method & Relative Cost ($\downarrow$) & Mean Constraint Violation ($\downarrow$) \\ \hline
         Ours         & 1.011      & 0.0002 \\ \hline
         No Dual \eqref{eq:heuristic-upper-level-opt}    & 1.014      & 0.002 \\ \hline
    \end{tabular}
    \caption{Constrained LQR Results}
    \label{tab:clqr-result}
\end{table}

\begin{figure}[]
    \centering
    \includegraphics[width=0.8\textwidth, trim={0.7cm 0 2cm 1.7cm},clip]{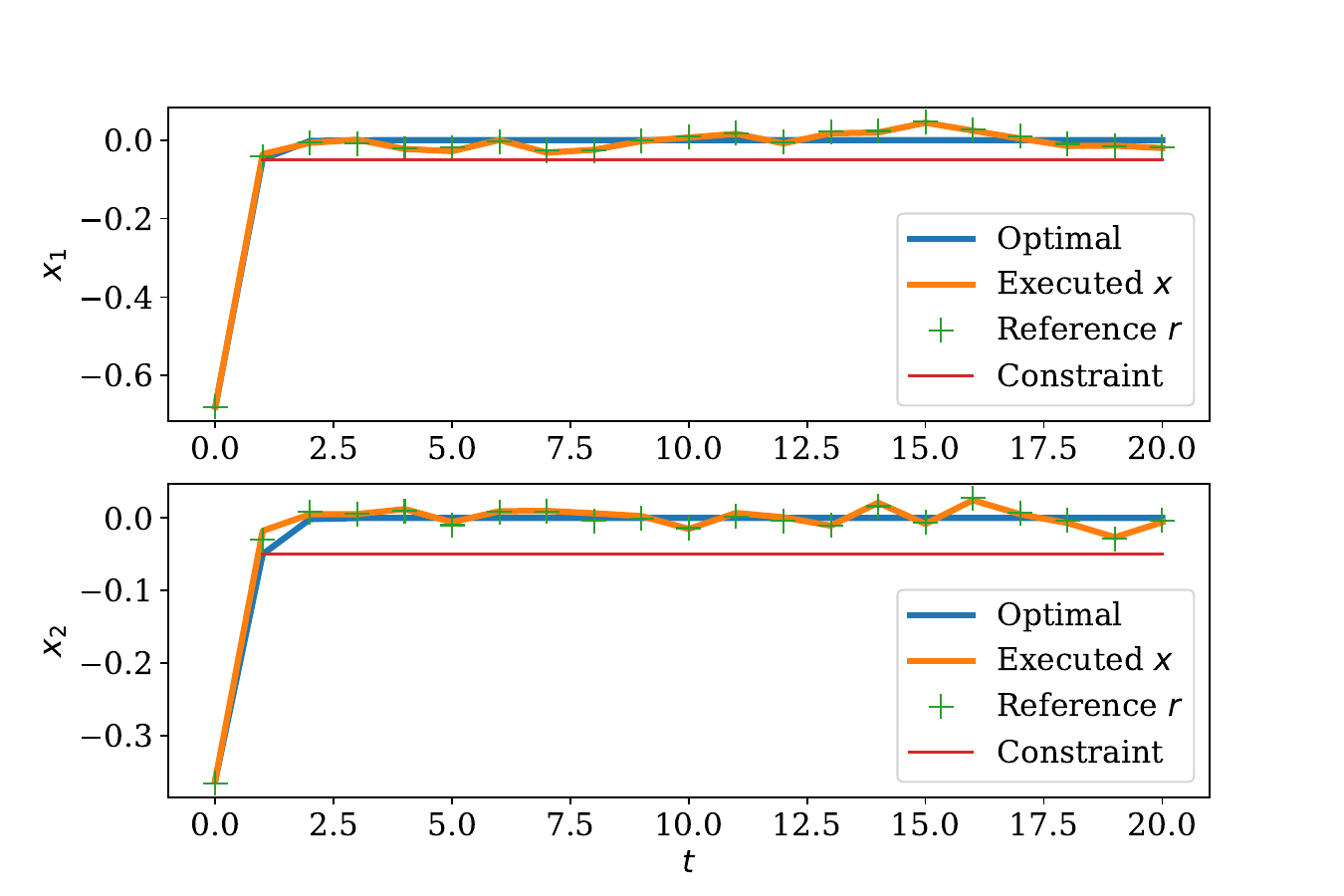}
    \caption{A Representative Sample Trajectory for Constrained LQR.}
    \label{fig:clqr-traj}
\end{figure}

As seen in Table \ref{tab:clqr-result} and the sample trajectories Figure \ref{fig:clqr-traj}, we can learn to generate reference trajectories satisfying the constraints. The planned trajectory is well-adapted to the learned tracking controller so that the executed trajectory also avoids constraint violations. This shows empirically that our algorithm can effectively learn to predict the dual variable even when the desired dual map is nonlinear. We again compare the results with solving for the reference without learning a dual network \eqref{eq:heuristic-upper-level-opt}, and observe that learning the dual network results in better coordination between the planner and the tracking controller. As a result, the approach with dual learning achieves better constraint satisfaction rates. We conclude this subsection by noting that in practice, one can tighten the constraints $x \in \mathcal{X}$ to ensure constraint satisfaction, even when there is tracking error. How to leverage the learned dual network to inform constraint tightening is an interesting direction of future work.

\vspace{-2mm}
\subsection{Unicycle}
Finally, we apply our algorithm to controlling a nonlinear unicycle system with state and control input
\begin{equation*}
    x_t = \begin{bmatrix} p_{x,t} \\ p_{y,t} \\ \theta_t \\ v_t  \end{bmatrix} \in \R^4,\quad u_t = \begin{bmatrix} a_t \\ \omega_t \end{bmatrix} \in \R^2,
\end{equation*}
where $p_x, p_y$ are the $x$ and $y$ positions, $\theta$ the heading angle, and $v$ the velocity of the unicycle. The two control inputs are the acceleration $a$ and the angular velocity (steering) $\omega$. We consider the discrete-time nonlinear dynamics given by
\begin{equation*}
    x_{t+1} = f(x_t, u_t) = \begin{bmatrix} p_{x,t} \\ p_{y,t}  \\ \theta_t \\ v_t \end{bmatrix} + 0.1
    \begin{bmatrix}
    \cos(\theta_t) v_t \\
    \sin(\theta_t) v_t \\
    \omega_t \\
    a_t  \end{bmatrix}.
\end{equation*}
We consider the problem of steering the vehicle to the origin, specified by the quadratic objective \eqref{eq:quadratic-cost} with $Q = \mathrm{diag}([1, 1, 0, 0])$, and $R = 0.01I_2$. The initial condition $\xi$ is sampled uniformly on the unit circle. We take $T=20$. The trajectory planner $\pi^\mathrm{plan}$ learns to generate references only for the positions $(p_x, p_y)$ instead of the full state.

The nonlinearity of the dynamics presents several challenges. First, we can no longer assume the form of the optimal tracking controller and its value function and have to parameterize both as neural networks. As a result of this non-convex parameterization of $p^\pi$, the reference generation problem \eqref{eq:traj-gen-param} becomes nonconvex. We use gradient descent to find reference trajectories that are locally optimal for the trajectory planning problem. Secondly, the nonlinear nature of the dynamics makes the learning of a tracking controller more difficult. To address this, we warmstart the tracking controller by training on simple line trajectories before running Algorithm \ref{alg:bi-actor-critic} in full with reference trajectory generated by solving \eqref{eq:traj-gen-param}.
This overcomes the difficulty that \eqref{eq:traj-gen-param} tends to generate bad trajectories when $p^\pi$ is randomly initialized. We train the tracking controller on simple references for $100,000$ transitions ($5,000$ episodes) as a warmstart, and then run Algorithm \ref{alg:bi-actor-critic} for $500,000$ transitions ($25,000$ episodes). We run the experiment both with and without training the dual network and report our results in Table \ref{tab:unicycle-result}. To make the result interpretable, we normalize the cost against iLQR as a baseline.\footnote{For each initial condition, we run iLQR with two random dynamically feasible initial trajectories. We take the lesser cost as iLQR's cost.}

First, we see that our learned policy achieves performance comparable to that of iLQR---we however emphasize that our policy is trained without explicit knowledge of the dynamics of the system. We note that the costs achieved by the policy learned with and without a dual network are similar. This could be due to the the trajectory generation problem \eqref{eq:traj-gen-param} not being solved exactly. However, learning with a dual network again leads to significantly better tracking performance, highlighting the importance of dual networks in coordinating the planning and tracking layers.

\begin{table}[h]
    \centering
    \begin{tabular}{|c|c|c|}
         \hline Method & Relative Cost ($\downarrow$) & Mean Tracking Deviation ($\downarrow$) \\ \hline
         iLQR & 1 & -\\ \hline
         Ours & 1.04 & 0.02 \\ \hline
         No Dual & 1.04 & 0.05 \\ \hline
    \end{tabular}
    \caption{Unicycle Results}
    \label{tab:unicycle-result}
    \vspace{-5mm}
\end{table}
\section{Conclusion}\label{sec:discussion}
We proposed a principled way of parameterizing and learning a layered control policy composed of a trajectory planner and a tracking controller. We derived our parameterization from an optimal control problem and showed that a dual network emerges naturally to coordinate the two components. We showed that our algorithm can learn to predict the optimal dual variable for unconstrained LQR problems and validated this theory via simulation experiments. Further simulation experiments also demonstrated the potential of applying this method to nonlinear control problems. Future work will explore using the dual network to inform constraint tightening and parameterizing the planner \eqref{eq:traj-gen-param} directly as a neural network to reduce online computation.


\section*{ACKNOWLEDGMENTS}
This work is supported in part by NSF award ECCS-2231349, SLES-2331880, and CAREER-2045834.
F. Yang would like to thank Bruce Lee, Thomas Zhang, and Kaiwen Wu for their helpful discussions.

\bibliographystyle{unsrtnat}
\bibliography{bibFiles/layered}

\newpage
\appendix
\section{Experiment Setup} \label{sec:appendix-DDPG}
\subsection{Hyperparameters for the Unconstrained LQR Experiments}
See Table \ref{tab:lqr-hyperparam}.
\begin{table}[h]
    \centering
    \begin{tabular}{|c|c|}
        \hline \textbf{Parameter} & \textbf{Value} \\ \hline
        TD3 Policy Noise        & 5e-4 \\
        TD3 Noise Clip          & 1e-3 \\
        TD3 Exploration Noise   & 0  \\
        actor learning rate     & 3e-3 \\
        actor batch size        & 256 \\
        critic learning rate    & 3e-3 \\
        critic batch size       & 256 \\
        dual learning rate      & 0.1 \\
        dual batch size         & 5 \\\hline
    \end{tabular}
    \caption{Hyperparameters for the Unconstrained LQR Experiments}
    \label{tab:lqr-hyperparam}
\end{table}

\subsection{Hyperparameters for the Constrained LQR Experiments}
See Table \ref{tab:clqr-hyperparam}.
\begin{table}[h]
    \centering
    \begin{tabular}{|c|c|}
        \hline \textbf{Parameter} & \textbf{Value} \\ \hline
        TD3 Policy Noise        & 5e-4 \\
        TD3 Noise Clip          & 1e-3 \\
        TD3 Exploration Noise   & 0  \\
        actor learning rate     & 3e-3 \\
        actor batch size        & 256 \\
        critic learning rate    & 3e-3 \\
        critic batch size       & 256 \\
        dual learning rate      & 3e-4 \\
        dual batch size         & 40 \\\hline
    \end{tabular}
    \caption{Hyperparameters for the Constrained LQR Experiments}
    \label{tab:clqr-hyperparam}
\end{table}
The dual network is chosen to be an MLP with one hidden layer of 128 neurons. The activation is chosen to be ReLU.

\subsection{Hyperparameters for the Unicycle Experiments}
See Table \ref{tab:unicycle-hyperparam}.
\begin{table}[h]
    \centering
    \begin{tabular}{|c|c|}
        \hline \textbf{Parameter} & \textbf{Value} \\ \hline
        TD3 Policy Noise        & 1e-3 \\
        TD3 Noise Clip          & 1e-2 \\
        TD3 Exploration Noise   & 6e-2  \\
        actor learning rate     & 1e-3 \\
        actor batch size        & 256 \\
        critic learning rate    & 1e-3 \\
        critic batch size       & 256 \\
        dual learning rate      & 5e-3 \\
        dual batch size         & 60 \\\hline
    \end{tabular}
    \caption{Hyperparameters for the Unicycle Experiments}
    \label{tab:unicycle-hyperparam}
\end{table}
The dual network is chosen to be an MLP with one hidden layer of 128 neurons. The actor and critic are both MLPs with a single hidden layer of 256 neurons. All activation functions are ReLU.
\section{Proofs for Theorem \ref{thm:opt-tracking}}\label{sec:appendix-proof-opt}
The LQR problem considered in Section \ref{sec:analysis} with costs \eqref{eq:quadratic-cost} and dynamics \eqref{eq:linear-dynamics} admits closed-form solutions to the $r$-update \eqref{eq:opt-traj-opt-update} and $x$-update \eqref{eq:opt-tracking-update}. In this section, we begin by showing that for this problem, the dual variable can indeed be written as a linear map of the initial condition $\xi$. We then derive the closed-form solutions to the updates \eqref{eq:opt-traj-opt-update} and \eqref{eq:opt-tracking-update}. Finally, we use a contraction argument to show our desired result in Theorem \ref{thm:opt-tracking}. In the process, we make clear the conditions on step size $\eta$ and batch size $B$ to guarantee the contraction.

For easing notation, for the rest of this section, we again define $\tilde{r} := r + \nu$. We also define the matrices
\begin{equation*}
    E := \begin{bmatrix}
    0 & 0  & &\hdots& 0\\
    B & 0  & &\hdots & 0\\
    AB & B & &\hdots & 0\\
    \vdots &\vdots & & & \vdots\\
    A^{T-1}B & A^{T-2}B  & &\hdots & B
    \end{bmatrix},\qquad
    F := \begin{bmatrix}
        I \\ A \\ A^2 \\ \vdots \\ A^T
    \end{bmatrix}.
\end{equation*}

\begin{lemma}\label{lem:dual-existence}
For the problem considered in section \ref{eq:redundant-prob}, given the initial condition $\xi$, the optimal dual variable can be expressed as a unique linear map from $\xi$ as
\begin{equation*}
    \nu^* = -\frac{2}{\rho} Q \left(-E(E^\top Q E + R)^{-1} E^\top Q F + F \right)\xi.
\end{equation*}
\end{lemma}
\begin{proof}
    From the KKT condition for the optimization problem \eqref{eq:redundant-prob}, we have that
    \begin{equation*}
        \nabla_r\left(r^\top Q r + \frac{\rho}{2}\norm{r + \nu^* - x^*}_2^2\right)\Big\vert_{r=r^*} = 0.
    \end{equation*}
    Solving for $r^*$, we get that
    \begin{equation*}
        2(Q + \frac{\rho}{2}I)r^* = -\rho \nu^* + \rho x^*.
    \end{equation*}
    Also from the KKT condition, we have that $r^*=x^*$. Subbing this into the above expression and rearranging terms, we get that
    \begin{equation*}
        \nu^* = -\frac{2}{\rho} Q r^* = -\frac{2}{\rho} Q x^*.
    \end{equation*}
    Finally, from the equivalence of the original problem \eqref{eq:original-prob} and the redundant problem \eqref{eq:redundant-prob}, we see that $x^*$ can be expressed in closed form as
    \begin{equation*}
        x^* = \left(-E(E^\top Q E + R)^{-1} E^\top Q F + F \right)\xi.
    \end{equation*}
\end{proof}

We now derive the closed-form update rules and show the following.
\begin{lemma}\label{lem:rx-expression}
    In the LQR setting, we have that the difference between the updates $r^+$ and $x^+$ can be written as a linear map of the initial condition $\xi$ as
\begin{equation*}
    r^+ - x^+ = H \Theta^{(k)} \xi + G \xi,
\end{equation*}
where
\begin{equation*}
\begin{aligned}
    H &:= -\frac{2}{\rho} P(Q + P)^{-1}P - \frac{\rho}{2} E\bar{R}^{-1} E^\top,\\
    G &:= HF + F.
\end{aligned}
\end{equation*}
$H$ is symmetric negative definite.
\end{lemma}

\begin{proof}
We start by deriving the closed-form expressions for solving both the updates \eqref{eq:opt-traj-opt-update} and \eqref{eq:opt-tracking-update}. We begin by writing out the update rule more explicitly. First, we note that we can write all $x, u$ satisfying the dynamics constraint as
\[x = Eu + F\xi,\quad u_t \in \R^{d_u}.\]
Define
\begin{align*}
    p(\tilde{r}, \xi, u) &= u^\top R u + \frac{\rho}{2} \norm{\tilde{r} - x}_2^2\\
    &= u^\top R u + \frac{\rho}{2} \norm{\tilde{r} - (Eu + F\xi))}_2^2.
\end{align*}
We can solve for the optimal control action in closed form as
\begin{align*}
    u^* &= \text{argmin}_u p(\tilde{r}, \xi, u)\\
    &= -\frac{\rho}{2} \bar{R}^{-1} E^\top (F\xi - \tilde{r}),
\end{align*}
where we defined
\begin{equation*}
    \bar{R} = R + \frac{\rho}{2} E^\top E.
\end{equation*}
Subbing this back, we get that
\begin{align*}
    p^*(\tilde{r}, \xi) &= p(\tilde{r}, \xi, u^*)\\
    &=(F\xi - \tilde{r})^\top (\frac{\rho}{2}I - \frac{\rho^2}{4}E\bar{R}^{-1}E^\top)(F\xi - \tilde{r})\\
    &=:(F\xi - \tilde{r})^\top P (F\xi - \tilde{r}),
\end{align*}
where we defined
\begin{equation*}
    P := \frac{\rho}{2}I - \frac{\rho^2}{4}E\bar{R}^{-1}E^\top.
\end{equation*}
Note that we arrived at $p^*$ from a partial minimization on $u$, which preserves the convexity of the problem on $r$. Thus, we have that
\begin{equation*}
    P \succ 0.
\end{equation*}
With the knowledge of $p^*$, we can now solve for $r$ in closed-form.
Since $Q + P \succ 0$, we have that
\begin{align*}
    r^+ &= \text{argmin}_r r^\top Q r + p^*(r + \nu, \xi)\\
    &= \text{argmin}_r r^\top Q r + (F\xi - r - \nu)^\top P (F\xi - r - \nu) \\
    &= (Q + P)^{-1} P (F\xi - \nu).
\end{align*}
Thus, overall, we update rules are given as
\begin{align*}
    r^{+} &= (Q + P)^{-1} P (F\xi - \nu),\\
    u^{+} &= -\frac{\rho}{2} \bar{R}^{-1} E^\top (F\xi - r^{(k+1)} - \vpred^{(k)}),\\
    x^{+} &= E u^{(k+1)} + F\xi.
\end{align*}

From the closed-form update rules specified above, we have that
\begin{align*}
    &r^+ - x^+ \\
    =& r^+ - (E(-\frac{\rho}{2} \bar{R}^{-1} E^\top (F\xi - r^{(k+1)} - \vpred^{(k)})) + F\xi)\\
    =& (I - \frac{\rho}{2} E\bar{R}^{-1} E^\top) r^+ + \frac{\rho}{2} E\bar{R}^{-1} E^\top (F\xi - \nu) - F\xi\\
    =& -(\frac{2}{\rho} P(Q + P)^{-1}P + \frac{\rho}{2} E\bar{R}^{-1} E^\top) \nu - (HF + F)\xi
\end{align*}
Denote
\begin{equation*}
\begin{aligned}
    H &:= -\frac{2}{\rho} P(Q + P)^{-1}P - \frac{\rho}{2} E\bar{R}^{-1} E^\top\\
    G &:= HF + F,
\end{aligned}
\end{equation*}
we get the expression that we desire.
From the fact that $P \succ 0$, $Q \succeq 0$, and $\bar{R} \succ 0$, it follows that $H \prec 0$.
\end{proof}

Now, recall that for any $\xi \in \R^{d_x}$, the optimal dual map $\Theta^*$ satisfies that $\nu^*=\Theta^*\xi$, where $\nu^*$ is the optimal dual variable for the given $\xi$. From the KKT condition of \eqref{eq:redundant-prob}, we know that $\nu^*$ induces a fixed point to the update rules \eqref{eq:opt-traj-opt-update} and \eqref{eq:opt-tracking-update}. Thus,
\begin{equation*}
    r^+ - x^+ = r^* - x^* = H \Theta^* \xi + G \xi = 0.
\end{equation*}
Since this holds for all $\xi \in \R^{d_x}$, we have that
\begin{equation*}
     H \Theta^* + G = 0.
\end{equation*}

Before starting with the main proof, we present the following lemma.
\begin{lemma}\label{lem:expectation}
For a set of i.i.d normal vectors $\xi_i \sim \mathcal{N}(0, I)$, we have that
\begin{equation*}
    \mathbb{E}\norm{\sum_{i=1}^{B}\frac{1}{B}\xi_i\xi_i^\top - I}_F \leq \sqrt{\frac{2d_x}{B}}.
\end{equation*}
\end{lemma}

\begin{proof}
    We begin by rewriting the above expression with a change of variables, where we define
    \begin{equation*}
        X := \sum_{i=1}^{B}\frac{1}{B}\xi_i\xi_i^\top.
    \end{equation*}
    Since $X$ is a sum of outer products of independently distributed normal random vectors, $X$ follows the Wishart distribution. Specifically, we have that
    \begin{equation*}
        X \sim W_{d_x}(I, B).
    \end{equation*}
    The above expression can then be bounded as
    \begin{equation*}
        \begin{aligned}
            \mathbb{E}\norm{\sum_{i=1}^{B}\frac{1}{B}\xi_i\xi_i^\top - I}_F
            &= \frac{1}{B} \mathbb{E}\norm{X - BI}_F\\
            &\leq \frac{1}{B} \sqrt{ \mathbb{E}\norm{X - BI}_F^2 }\\
            &= \frac{1}{B} \sqrt{ \Tr \mathbb{E}(X - BI)(X - BI)^\top}\\
            &= \frac{1}{B} \sqrt{ \Tr \Var(X)}\\
            &= \frac{1}{B} \sqrt{2d_x B}\\
            &= \sqrt{\frac{2d_x}{B}},
        \end{aligned}
    \end{equation*}
    where we first used Jensen's inequality and then, in the following equalities, used the properties of Wishart random variables.
\end{proof}

We can now start analyzing the progress of the dual update. First, combining Lemma \ref{lem:rx-expression} with the dual update rule \eqref{eq:lqr-dual-update}, we have that
\begin{equation}\label{eq:app-Theta-update}
    \begin{aligned}
    \Theta^{(k+1)} &= \Theta^{(k)} + \eta \sum_{i=0}^{B}\frac{1}{B} (H \Theta^{(k)} \xi_i + G \xi_i) \xi_i^\top\\
    &= \Theta^{(k)} + \eta (H \Theta^{(k)} + G) \sum_{i=0}^{B}\frac{1}{B}  \xi_i \xi_i^\top.
\end{aligned}
\end{equation}

Thus, for the expected norm of interest, we have
\begin{equation*}
    \begin{aligned}
        &\mathbb{E} \norm{ \Theta^{(k+1)} - \Theta^* }\\
        =& \mathbb{E} \norm{ \Theta^{(k)} + \eta (H \Theta^{(k)} + G) \sum_{i=0}^{B}\frac{1}{B} \xi_i \xi_i^\top - \Theta^* }\\
        \overset{(a)}{=}&\mathbb{E} \norm{\Theta^{(k)} - \Theta^* + \eta (H \Theta^{(k)} + G) \sum_{i=0}^{B}\frac{1}{B}  \xi_i \xi_i^\top - \eta (H \Theta^* + G) \sum_{i=0}^{B}\frac{1}{B} \xi_i \xi_i^\top }\\
        =&\mathbb{E} \norm{ \Theta^{(k)} - \Theta^* + \eta H (\Theta^{(k)} - \Theta^*) \sum_{i=0}^{B}\frac{1}{B} \xi_i \xi_i^\top }\\
        =&\mathbb{E} \norm{(I + \eta H) (\Theta^{(k)} - \Theta^*) + \eta H (\Theta^{(k)} - \Theta^*) (\sum_{i=0}^{B}\frac{1}{B} \xi_i \xi_i^\top - I) }\\
        \leq& \left(\norm{I + \eta H}_2 + \eta \norm{H}_2 \mathbb{E}\norm{\sum_{i=0}^{B}\frac{1}{B} \xi_i \xi_i^\top - I}\right) \norm{\Theta^{(k)}- \Theta^*}\\
        \leq& \left(\norm{I + \eta H}_2 + \eta \norm{H}_2 \sqrt{\frac{2d_x}{B}}\right)\norm{\Theta^{(k)}- \Theta^*},
    \end{aligned}
\end{equation*}
where in step $(a)$ we used the above fact that $H\Theta^* + G = 0$, and in the final step, we used Lemma \ref{lem:expectation}. We have the desired contraction if
\begin{equation*}
    \gamma(\eta, B) := \norm{I + \eta H}_2 + \eta \norm{H}_2 \sqrt{\frac{2d_x}{B}} < 1.
\end{equation*}
Note that choosing
$$\eta^* = \frac{2}{\smax(H)+\smin(H)}$$
minimizes the norm $\norm{I + \eta H}$. Solving for $B$ with this choice of $\eta$, we have that $B$ needs to satisfy that
\begin{equation*}
    B > \frac{2d_x \smax^2(H)}{\smin(H)^2}.
\end{equation*}
Following these appropriate choices of $\eta$ and $B$, we have
\begin{equation*}
    \mathbb{E} \norm{\Theta^{(k)} - \Theta^{*}} \leq \gamma^k \mathbb{E} \norm{\Theta^{(0)} - \Theta^{*}}.
\end{equation*}

\section{Proof of Theorem \ref{thm:subopt-tracking}}\label{sec:appendix-proof-subopt}
We start by showing that when $H$ and $G$ (See Lemma \ref{lem:rx-expression} above) are perturbed by small additive perturbations, the algorithm can still converge to the vicinity of the optimal $\Theta^*$ if the perturbations are small enough.
\begin{lemma}
    Consider the perturbations in $H$ and $G$ as
\begin{equation*}
    H' = H + \Delta_H, \quad G' = G + \Delta_G.
\end{equation*}
If the perturbation in $H$ satisfies that
\begin{equation*}
    \norm{\Delta_H} < \frac{1}{1 + \sqrt{\frac{2d_x}{B_0}}} \frac{\smax(H) - \smin(H)}{\smax(H) + \smin(H)},
\end{equation*}
for any $B_0 \geq 1$, then we can pick step size $\eta$ and batch size $B$ such that the update \eqref{eq:app-Theta-update} converges to the vicinity of the optimal dual variable, i.e., that
\begin{equation*}
    \E \norm{\Theta^{(k)} - \Theta^*} \leq
    \gamma^k \E \norm{\Theta^{(0)} - \Theta^*} + \frac{1-\gamma^k}{1-\gamma} e(\norm{\Delta_H}, \norm{\Delta_G}),
\end{equation*}
where $0 < \gamma < 1$, and
\begin{equation}\label{eq:delta-H-to-radius}
    e(\norm{\Delta_H}, \norm{\Delta_G}) := (1 + \sqrt{\frac{2d_x}{B}})\left(\norm{\Delta_G} + \norm{\Delta_H} \norm{\Theta^*}\right)
\end{equation}
\end{lemma}

\begin{proof}
We follow similar steps as when we showed the similar result for the unperturbed case.
\begin{equation*}
    \begin{aligned}
        &\mathbb{E} \norm{ \Theta^{(k+1)} - \Theta^* }\\
        =&\; \mathbb{E} \Bigg\lVert \Theta^{(k)}  - \Theta^* +\eta ((H+\Delta_H) \Theta^{(k)} + (G+\Delta_G)) \sum_{i=0}^{B}\frac{1}{B} \xi_i \xi_i^\top \Bigg\rVert\\
        =&\;\mathbb{E} \Bigg\lVert \Theta^{(k)} - \Theta^* +\eta (H \Theta^{(k)} + G + (\Delta_H \Theta^{(k)} + \Delta_G)) \sum_{i=0}^{B}\frac{1}{B}  \xi_i \xi_i^\top
        - \eta (H \Theta^* + G) \sum_{i=0}^{B}\frac{1}{B} \xi_i \xi_i^\top \Bigg\rVert\\
        \leq&\;\mathbb{E} \norm{\Theta^{(k)} - \Theta^* + \eta H (\Theta^{(k)} - \Theta^*) \sum_{i=0}^{B}\frac{1}{B} \xi_i \xi_i^\top} +\E\norm{(\Delta_H \Theta^{(k)} + \Delta_G)) \sum_{i=0}^{B}\frac{1}{B} \xi_i \xi_i^\top}\\
        \leq&\;\left(\norm{I + \eta H}_2 + \eta \norm{H}_2 \sqrt{\frac{2d_x}{B}}\right)\norm{\Theta^{(k)}- \Theta^*}_F +\E\norm{(\Delta_H \Theta^{(k)} + \Delta_G)) \sum_{i=0}^{B}\frac{1}{B} \xi_i \xi_i^\top},
    \end{aligned}
\end{equation*}
where the last step followed the same steps in the proof of Theorem \ref{thm:opt-tracking}. We now proceed to bound the second term.
\begin{align*}
    &\E\norm{(\Delta_H \Theta^{(k)} + \Delta_G) \sum_{i=0}^{B}\frac{1}{B} \xi_i \xi_i^\top}\\
    \leq& \norm{(\Delta_H (\Theta^{(k)}-\Theta^*) + (\Delta_G + \Delta_H \Theta^*))} \E\norm{\sum_{i=0}^{B}\frac{1}{B} \xi_i \xi_i}\\
    \leq& (1 + \sqrt{\frac{2d_x}{B}}) \Big(\norm{\Delta_H} \norm{\Theta^{(k)}-\Theta^*} + \norm{\Delta_G} + \norm{\Delta_H} \norm{\Theta^*} \Big)\\
\end{align*}
Combining this with the unperturbed result, we get that
\begin{align*}
    &\mathbb{E} \norm{ \Theta^{(k+1)} - \Theta^* }_F\\
    \leq &\left(\norm{I + \eta H}_2 + \eta \norm{H}_2 \sqrt{\frac{2d_x}{B}} + (1 + \sqrt{\frac{2d_x}{B}}) \norm{\Delta_H} \right) \norm{\Theta^{(0)} - \Theta^*} + e(\norm{\Delta_H}, \norm{\Delta_G}),
\end{align*}
where
\begin{align*}
    e(\norm{\Delta_H}, \norm{\Delta_G}) := (1 + \sqrt{\frac{2d_x}{B}})\left(\norm{\Delta_G} + \norm{\Delta_H} \norm{\Theta^*}\right)
\end{align*}
For the iterations to be contractive, we would need that
\begin{equation*}
    \left(\norm{I + \eta H}_2 + \eta \norm{H}_2 \sqrt{\frac{2d_x}{B}} + (1 + \sqrt{\frac{2d_x}{B}}) \norm{\Delta_H} \right) < 1.
\end{equation*}
Again, note that choosing
$$\eta^* = -\frac{2}{\lmax(H)+\lmin(H)}$$
minimizes the norm $\norm{I + \eta H}$. Thus, for the inequality to hold, $\norm{\Delta_H}$ needs to satisfy
\begin{equation*}
\begin{aligned}
    \left(1 + \sqrt{\frac{2d_x}{B_0}}\right)\norm{\Delta_H} < \norm{1 + \eta^* H} = \frac{\lmax(H) - \lmin(H)}{\lmax(H) + \lmin(H)},
\end{aligned}
\end{equation*}
for some $B_0\geq1$ or equivalently
\begin{equation*}
    \norm{\Delta_H} < \frac{1}{1 + \sqrt{\frac{2d_x}{B_0}}} \frac{\lmax(H) - \lmin(H)}{\lmax(H) + \lmin(H)}.
\end{equation*}
We can then pick
\begin{equation*}
    B > \text{max}\left(\frac{2d_x\eta \norm{H}_2^2}{(1 - \norm{I + \eta H})^2 - (1 + \sqrt{\frac{2d_x}{B_0}}) \norm{\Delta_H}}, B_0\right),
\end{equation*}
so that
\begin{equation*}
    \gamma := \norm{I + \eta H}_2 + \eta \norm{H}_2 \sqrt{\frac{2d_x}{B}} + (1 + \sqrt{\frac{2d_x}{B}}) \norm{\Delta_H} < 1.
\end{equation*}
The result then follows from telescoping the sum.
\end{proof}

We now consider the perturbations we described in Section \ref{sec:analysis-subopt} and bound the terms $\norm{\Delta_H}$ and $\norm{\Delta_G}$ in terms of the perturbations $\epsilon_P, \epsilon_{u,r}, \epsilon_{u,\xi}$.
\begin{lemma}
Consider the perturbations specified in \eqref{eq:r-update-perturb} and \eqref{eq:u-update-perturb}. Denote the norms of the perturbations as
\begin{equation*}
    \epsilon_P = \norm{\Delta_P}, \epsilon_{u,r} = \norm{\Delta_{u,r}}, \epsilon_{u,\xi} = \norm{\Delta_{u,\xi}}.
\end{equation*}
If $\epsilon_P < \frac{\lmin(Q+P)}{2}$, we have that $\norm{\Delta_H}$ and $\norm{\Delta_G}$ as
\begin{equation}\label{eq:dh-norm-bound}
    \begin{aligned}
    \norm{\Delta_H} &< e_H(\epsilon_P, \epsilon_{u, r}),\\
    &:= \frac{2\epsilon\norm{P}}{\lmin(Q+P)} + \frac{\epsilon^2}{\lmin(Q+P)} + 
        \frac{2\epsilon(\norm{P} + \epsilon)^2}{\lmin(Q+P)^2},
    \end{aligned}
\end{equation}
where $\epsilon = \text{max}(\epsilon_P, \frac{\rho}{2}\norm{E}\epsilon_{u,r})$.
\begin{equation}\label{eq:dg-norm-bound}
\begin{aligned}
    \norm{\Delta_G} &< e_G(\epsilon_P, \epsilon_{u, r}, \epsilon_{u, \xi})\\
    &:= \norm{F} e_H(\epsilon_P, \epsilon_{u, r}) + \epsilon_{u,\xi}.
\end{aligned}
\end{equation}
\end{lemma}
\begin{proof}
We first note that the perturbed update rule gives the updates
\begin{gather*}
    r^+ = (Q + P + \Delta_P)^{-1} (P + \Delta_P) (r + \nu - F\xi),\\
    \hat{u}(\tilde{r}, \xi) = u^*(\tilde{r}, \xi) + \Delta_{u,r}\tilde{r} + \Delta_{u, \xi} \xi.
\end{gather*}
The difference between $r^+$ and $x^+$ can be summarized as
\begin{equation*}
\begin{aligned}
    &r^+ - x^+ \\
    =& (I - \frac{\rho}{2} E\bar{R}^{-1} E^\top-E\Delta_{u,r})r^+ + \frac{\rho}{2} E\bar{R}^{-1} E^\top (F\xi - \nu) -(F+\Delta_{u, \xi})\xi\\
    =& \frac{2}{\rho}(P-\frac{\rho}{2}E\Delta_{u,r})r^+ + \frac{\rho}{2} E\bar{R}^{-1} E^\top (F\xi - \nu) -(F+\Delta_{u, \xi})\xi\\
    =& (H + \Delta_H)\nu - (G+\Delta_G)\xi\\
\end{aligned}
\end{equation*}
where
\begin{align*}
    \Delta_H &:= \frac{2}{\rho}(P-\frac{\rho}{2}E\Delta_{u,r})(Q + P + \Delta_P)^{-1} (P + \Delta_P) -\frac{2}{\rho}P(Q + P)^{-1}P\\
    \Delta_G &:= \Delta_H F + E\Delta_{u, \xi}.
\end{align*}
We now proceed to bound $\norm{\Delta_H}$. Denote the reduced SVD of $\Delta_P$ as
\begin{equation*}
    \Delta_P = U C V^\top.
\end{equation*}
For the sake of simplicity, we denote $D:= P + Q$. By Woodbury matrix identity, we have that
\begin{equation*}
     (A + UCV^\top)^{-1} \\
     = A^{-1} - A^{-1} U (C^{-1} + V^\top A^{-1} U)^{-1} V^\top A^{-1}.
\end{equation*}
Thus, we have that
\begin{align*}
    &(P-\frac{\rho}{2}E\Delta_{u,r})(Q + P + \Delta_P)^{-1} (P + \Delta_P)\\
    =& PD^{-1}P + \frac{\rho}{2}E\Delta_{u,r} D^{-1} P + PD^{-1}\Delta_P- \frac{\rho}{2}E\Delta_{u,r}D^{-1}\Delta_P+\\ &\qquad\qquad (P-\frac{\rho}{2}E\Delta_{u,r})
    D^{-1} U (C^{-1} + V^\top D^{-1} U)^{-1} V^\top D^{-1}(P + \Delta_P).
\end{align*}
The first term corresponds to the unperturbed $H$. We thus proceed to bound all the other terms left. Define
\begin{equation*}
    \epsilon = \text{max}(\epsilon_P, \frac{\rho}{2}\norm{E}\epsilon_{u,r}),
\end{equation*}
we have that
\begin{gather*}
    \norm{\frac{\rho}{2}E\Delta_{u,r} D^{-1} P} \leq \frac{\rho\norm{E\Delta_{u,r}}\norm{P}}{2\lmin(D)} \leq \frac{\epsilon\norm{P}}{\lmin(D)},\\
    \norm{PD^{-1}\Delta_P}\leq\frac{\norm{\Delta_P}\norm{P}}{\lmin(D)} \leq \frac{\epsilon\norm{P}}{\lmin(D)},\\
    \norm{\frac{\rho}{2}E\Delta_{u,r}D^{-1}\Delta_P} \leq \frac{\rho\norm{E\Delta_{u,r}}\norm{\Delta_P}}{2\lmin(D)} \leq \frac{\epsilon^2}{\lmin(D)}.
\end{gather*}
To bound the last term, we use the fact that
\begin{align*}
    &\Big\lVert(P-\frac{\rho}{2}E\Delta_{u,r}) D^{-1} U (C^{-1} + V^\top D^{-1} U)^{-1} V^\top D^{-1} (P+\Delta_{P})\Big\rVert\\
    \leq &(\norm{P}+\epsilon)^2 \cdot \frac{1}{\lmin(D)}\cdot \frac{1}{\smin(C^{-1} + V^\top D^{-1}U)}
\end{align*}
We invoke the reverse triangle inequality to get that
\begin{align*}
    \frac{1}{\smin(C^{-1} + V^\top D^{-1}U)} &\leq \frac{1}{\smin(C^{-1}) - \norm{V^\top D^{-1}U}}\\
    &= \frac{1}{\frac{1}{\norm{C}} - \frac{1}{\lmin(D)}}\\
    &\leq \epsilon_P\cdot\frac{\lmin(D)}{\lmin(D) - \epsilon_P}
\end{align*}
From the assumption that $\epsilon_P < \frac{\lmin(D)}{2}$, we have that
\begin{equation*}
    \frac{\lmin(D)}{\lmin(D) - \epsilon_P} \leq 2.
\end{equation*}
Thus, we have that
\begin{align*}
    &\Big\lVert(P-\frac{\rho}{2}E\Delta_{u,r}) D^{-1} U (C^{-1} + V^\top D^{-1} U)^{-1} V^\top D^{-1} (P+\Delta_{P})\Big\rVert\\
    \leq &\frac{2\epsilon(\norm{P} + \epsilon)^2}{\lmin(D)^2}. 
\end{align*}
Thus, we overall have that
\begin{equation*}
    \norm{\Delta_H} \leq 
    \frac{2\epsilon\norm{P}}{\lmin(D)} + \frac{\epsilon^2}{\lmin(D)} + \frac{2\epsilon(\norm{P} + \epsilon)^2}{\lmin(D)^2}
\end{equation*}
and that
\begin{equation*}
    \norm{\Delta_G} \leq \norm{\Delta_H}\norm{F} + \epsilon_{u, \xi}.
\end{equation*}
\end{proof}
Combining the two Lemmas above, we can now state Theorem \ref{thm:subopt-tracking} formally.
\begin{theorem} (Formal statement of Theorem \ref{thm:subopt-tracking})
Consider the cost functions \eqref{eq:quadratic-cost} and dynamics \eqref{eq:linear-dynamics}. Consider the update rules \eqref{eq:lqr-dual-pred}-\eqref{eq:lqr-dual-update} with the perturbations \eqref{eq:r-update-perturb} and \eqref{eq:u-update-perturb}. Denote the size of the perturbations as 
$$\epsilon_P = \norm{\Delta_P}, \epsilon_{u,r} = \norm{\Delta_{u,r}}, \epsilon_{u,\xi} = \norm{\Delta_{u,\xi}}.$$
Define $e_H$ as in \eqref{eq:dh-norm-bound} and $e_G$ as in \eqref{eq:dg-norm-bound}. Given any $\Theta^{(0)}$, if the perturbations $\epsilon_P, \epsilon_{u,r}$ satisfy that,
\begin{equation*}
     e_H(\epsilon_P, \epsilon_{u, r}) < \frac{1}{1 + \sqrt{\frac{2d_x}{B}}} \frac{\smax(H) - \smin(H)}{\smax(H) + \smin(H)},
\end{equation*}
for any $B_0 \geq 1$, one can pick
\begin{equation*}
    \eta = \frac{2}{\smax(H)+\smin(H)}
\end{equation*}
and batch size
\begin{equation*}
        B > \text{max}\left(\frac{2d_x\eta \norm{H}_2^2}{(1 - \norm{I + \eta H})^2 - (1 + \sqrt{\frac{2d_x}{B_0}}) e_H(\epsilon_P, \epsilon_{u,r})}, B_0\right),
\end{equation*} such that
\begin{equation*}
    \E \norm{\Theta^{(k)} - \Theta^*} \leq \\
    \gamma^k \E \norm{\Theta^{(0)} - \Theta^*} + \frac{1-\gamma^k}{1-\gamma} e(\epsilon_P, \epsilon_{u,r}, \epsilon_{u,\xi}),
\end{equation*}
where $0 < \gamma < 1$, and
\begin{equation*}
e(\epsilon_P, \epsilon_{u,r}, \epsilon_{u,\xi}) := (1 + \sqrt{\frac{2d_x}{B}})\\
\left(e_G(\epsilon_P, \epsilon_{u, r}, \epsilon_{u, \xi}) + e_H(\epsilon_P, \epsilon_{u, r}) \norm{\Theta^*}\right).
\end{equation*}
\end{theorem}
We verify the predictions of the theorem qualitatively in the experiment section.
\section{Planning Only Subset of States} \label{sec:subset-plan}
We consider the case where the state cost $\mathcal{C}$ and constraints $\mathcal{X}$ only require a subset of the states, i.e., if they are defined in terms of $z_t = g(x_t) \in \R^{d_z}$, with $d_z < d_x$. Specifically, we consider the problem
\begin{equation*}
    \begin{aligned}
        \underset{x, u, z}{\text{minimize}}&\quad
            \mathcal{C}(z) + \mathcal{D}(u)\\
        \text{subject to}&\quad x_{t+1} = f(x_t, u_t), \quad \forall t=0, 1, ..., T-1,\\
        &\quad z\in \mathcal{X}, \quad u \in \mathcal{U},\\
        &\quad x_0 = \xi,\\
        &\quad z_t = g(x_t).
    \end{aligned}
\end{equation*}
In this case, one can modify the redundant constraint to be $r_t = z_t$ to arrive at the following redundant problem
\begin{equation*}
    \begin{aligned}
        \underset{r, x, u}{\text{minimize}}&\quad
            \mathcal{C}(r) + \mathcal{D}(u)\\
        \text{subject to}&\quad x_{t+1} = f(x_t, u_t), \quad \forall t=0, 1, ..., T-1,\\
        &\quad r\in \mathcal{X}, \quad u \in \mathcal{U},\\
        &\quad x_0 = \xi, \quad r = g(x),
    \end{aligned}
\end{equation*}
where we wrote $g(x)$ to denote $ [g(x_0), g(x_1), ..., g(x_T)]^\top$ with a slight abuse of notation. A similar derivation to that in Section \ref{sec:layering} then arrives at the following iterative update
\begin{align*}
    &
    \begin{aligned}
        \quad (r^+, x^+, u^+) =\underset{r, x, u}{\text{argmin}}\;&\mathcal{C}(r) + \mathcal{D}(u) + \frac{\rho}{2} \lVert r + \nu - g(x) \rVert_2^2\\
        \text{s.t.}\;& x_{t+1} = f(x_t, u_t),\quad\forall t,\\
        &r\in \mathcal{X}, u\in\mathcal{U},\\
        & x_0 = \xi
    \end{aligned}\\
    &\quad \nu^+ = \nu + (r^+ - x^+)
\end{align*}
and the nested optimization
\begin{equation*}
    \begin{aligned}
        r^+ = \underset{r}{\text{minimize}}&\quad\mathcal{C}(r)+ p^*(r + \nu; \xi)\\
        \text{s.t.}&\quad r\in \mathcal{X}
    \end{aligned}
\end{equation*}
where $p^*(r +\nu; \xi)$ is the locally optimal value of the $(x, u)$-minimization step
\begin{equation*}
    \begin{aligned}
        p^*(r+\nu; \xi) = \underset{x, u}{\text{min}}&\quad\mathcal{D}(u) + \frac{\rho}{2} \lVert r + \nu - g(x) \rVert_2^2\\
        \text{s.t.}&\quad x_{t+1} = f(x_t, u_t),\; u \in \mathcal{U},\; \forall t,\\
        &\quad x_0 = \xi.
    \end{aligned}
\end{equation*}
Note that the only difference is that the trajectory planner only generates reference trajectories on the states $z$ required for the state cost and constraints, and that the tracking cost for the lower-level controller also only concerns tracking those states.

\end{document}